\documentclass[1p,final]{elsarticle}
\usepackage{amsfonts,natbib,pslatex}
\usepackage{latexsym,amsmath,amsthm,amssymb,amsxtra,mathrsfs,times,CJK}
\usepackage{makecell}
\usepackage{color}
\usepackage{amsfonts,color}
\usepackage{amssymb,amsmath,latexsym}
\usepackage{amssymb}
\usepackage{extarrows}
\biboptions{sort&compress}

\usepackage[para]{threeparttable}

\usepackage{tikz}

\usepackage[colorlinks,linkcolor=blue,anchorcolor=blue,citecolor=blue]{hyperref}

\usepackage[para]{threeparttable}

\newtheorem{theorem}{Theorem}[section]
\newtheorem{lemma}[theorem]{Lemma}

\newcommand{\f}{{\mathbb{F}}}

\newcommand{\F}{{\mathcal{F}}}
\newcommand{\G}{{\mathcal{G}}}
\newcommand{\C}{{\mathcal{C}}}
\newcommand{\N}{{\mathrm{Null}}}
\journal{Finite Fields and Their Applications}
\allowdisplaybreaks[1]
\begin{document}

	\begin{frontmatter}
		
		
		
		\title{On the parameters of extended primitive cyclic codes and the related designs}

		\author[SWJTU]{Haode Yan\corref{cor1}}
		\ead{hdyan@swjtu.edu.cn}
		
			\author[SWJTU]{Yanan Yin}
		\ead{yinyn5050@163.com}

		\cortext[cor1]{Corresponding author}
		\address[SWJTU]{School of Mathematics, Southwest Jiaotong University, Chengdu, 610031, China}

		\begin{abstract}Very recently, Heng et al. studied a family of extended primitive cyclic codes. It was shown that the supports of all codewords with any fixed nonzero Hamming weight of this code supporting $2$-designs. In this paper, we study this family of extended primitive cyclic codes in more details. The weight distribution is determined. The parameters of the related $2$-designs are also given. Moreover, we prove that the codewords with minimum Hamming weight supporting $3$-designs, which gives an affirmative solution to Heng's conjecture.
		
		\end{abstract}

		\begin{keyword}
			Linear code\sep Cyclic code \sep Extended primitive cyclic code \sep $t$-design	\sep Linearized polynomial
			\MSC  94B05 \sep 94A05

		\end{keyword}
		
	\end{frontmatter}

\section{Introduction}
 Let $q$ be a prime power, and $\f_q$ be the finite field with $q$ elements. A $q$-ary linear $[n,k,d]$ code $\mathcal{C}$ is a $k$-dimension linear subspace over $\f_q$ with  the minimum distance $d$. For a linear $[n,k]$ code $\mathcal{C}$, any $k\times n$ matrix $G$ whose rows form a basis for $\mathcal{C}$ is called a generator matrix of $\mathcal{C}$. A $q$-ary linear code $\mathcal{C}$ of length $n$ is called cyclic if $(c_0,c_1,\cdots,c_{n-1}) \in \mathcal{C}$ implies $(c_{n-1},c_0,\cdots,c_{n-2})\in \mathcal{C}$. When the length of a $q$-ary cyclic code is $q^m-1$ for some positive ingeter $m$, the cyclic code is called primitive. Let $A_i$ $(i=1,2,\cdots,n)$ denote the number of codewords with Hamming weight $i$ in $\mathcal{C}$. The weight enumerator of $\mathcal{C}$ is defined by $1+A_1z+A_2z^2+\dots+A_nz^n$ and the sequence $(1,A_1,A_2,\cdots,A_n)$ is referred to as the weight distribution of $\mathcal{C}$. The weight distribution of $\mathcal{C}$ shows the minimum distance and the error
  correcting capability of $\mathcal{C}$. To determine the weight distribution of a linear code has always been a hot topic in recent years \cite{Ding2016,Dinh2015,Luo2008,Feng2007,Heng2016,Li2014,Xiong2016,Zhou2014,Zeng2010}.
 
 Let $\mathcal{P}$ be a set with $n$ elements, and let $\mathcal{B}$ be a set of $k$-subsets of $\mathcal{P}$, where $n$ and $k$ are positive integers with $1\leq k\leq n$. Let $t$ be a positive integer with $t \leq k$. A pair $\mathbb{D}=(\mathcal{P},\mathcal{B})$ is called a $t$-$(n,k,\lambda)$ design, or simply $t$-design, which refers to every $t$-subset of $\mathcal{P}$ is contained in exactly $\lambda$ elements of $\mathcal{B}$. The elements of $\mathcal{P}$ are called points, and those of $\mathcal{B}$ are referred to as blocks. A $t$-design is called simple if there are no repeated blocks in $\mathcal{B}$. According to \cite{PLess2003}, the complementary design of $(\mathcal{P},\mathcal{B})$ is the design $(\mathcal{P},\mathcal{B}^{\prime})$, where $\mathcal{B}^{\prime}$ consists of the complements of the blocks in $\mathcal{B}$ and  $(\mathcal{P},\mathcal{B}^{\prime})$ is in fact a $t$-design. It is well known that if the pair $(\mathcal{P},\mathcal{B})$ is a simple $t$-$(n,k,\lambda)$ design, we will obtain the following relation:
 \begin{equation}
 	\label{eqn-1}\binom{n}{t}\lambda=\binom{k}{t}|\mathcal{B}|,
 \end{equation}
where $|\mathcal{B}|$ denotes the number of blocks in $\mathcal{B}$.

Linear codes and $t$-designs are companions. For a linear code $\mathcal{C}$ of length $n$, let $\mathcal{P}=\{1,2,\cdots,n\}$ as the set of coordinate. For any codeword $\mathbf{c}=(c_1,c_2,\cdots,c_n) \in \mathcal{C}$, the support of $\mathbf{c}$ is defined as supp($\mathbf{c}$)=$\{1\leq i\leq n:c_i\neq 0\}$. Let $\mathcal{B}_k$ denoted the set of supports of all codewords with Hamming weight $k$ in $\mathcal{C}$. If the pair $(\mathcal{P},\mathcal{B}_k)$ is a $t$-$(n,k,\lambda)$ design for some positive integers $t$ and  $\lambda$, we say that the supports of codewords with Hamming weight $k$ in $\mathcal{C}$ hold or support a $t$-$(n,k,\lambda)$ design. Using linear codes to construct $t$-designs has been documented in many literatures \cite{Ding2018,Ding2018-2,Ding2019,Ding2015,Ding2017,Ding2020,Ding2021,Du2020,Heng2020,Heng2023,Liu2022,Liu2021,Tang2021,Tang2019,Wang2023,Xiang2022,Xu2022,Xiang2020,Xiang2022-2,Yan2022}.

In \cite{Heng2023}, Heng et al. constructed a family of
extended primitive cyclic codes. The parameters of this family of codes and the related $t$-designs are also studied. More precisely, it was proved that the codewords with any fixed nonzero weight in this code hold $2$-designs. In this paper, we study this family of cyclic codes and the related $t$-designs in more details. We determine the weight distribution of the code and investigate parameters of the $2$-designs. Besides, we prove that the minimum weight codewords of this code support a simple $3$-design when $p=2$, which confirms a conjecture proposed in \cite{Heng2023}.
 
The rest of this paper is organized as follows. In Section \ref{sec2}, we introduce the notation and present some auxiliary results which will be used later. The parameters of the codes and the related designs are given in Section \ref{sec3} and Section \ref{sec4}, respectively. Section \ref{sec5} concludes this paper.
\section{Preliminaries}
\label{sec2}
We begin this section by fixing some notation which will be used throughout this paper unless otherwise stated. Then we introduce some basic results, which will be used in the sequel.
\subsection{Notation}

\begin{itemize}
	\item $h$ and $m$ are positive integers with $h<m$.
	\item $q=p^m$, where $p$ is a prime. $\f_q$ is the finite field with $q$ elements. Sometimes we consider $\f_q$ as an $m$-dimensional vector space over $\f_p$.
	\item $|S|$ denotes the cardinality of the set $S$. $S^c$ is the complementary of $S$. $S+a=\{s+a|s\in S\}$.
	\item $D_h$ is a $(h+2)\times q$ matrix over $\f_q$, which is defined in (\ref{matrix1}). $\mathcal{C}_{D_h}$ is the linear code generated by $D_h$.
	\item For a codeword $\mathbf c\in\mathcal{C}_{D_h}$, $\mathrm{wt}(\mathbf c)$ denotes the Hamming weight of $\mathbf c$. $A_k (0\leq k \leq q)$ is the number of codewords with Hamming weight $k$ in $\mathcal{C}_{D_h}$.
	\item $\mathrm{supp}(\mathbf c)$ is the support of $\mathbf c$, which is a subset of $\{1,2,\cdots,q\}$. In this paper, by the construction of $\mathcal{C}_{D_h}$, there is a bijection $\tau$ between $\{1,2,\cdots,q\}$ and $\f_q$. Namely, $\tau(i)=\alpha^i$ if $1\leq i \leq q-1$ and $\tau(q)=0$, where $\alpha$ is a fixed primitive element in $\f_q$. Sometimes we do not  distinguish the subset of $\{1,2,\cdots,q\}$ and the corresponding subset of $\f_q$.
	\item $\binom{n}{k}_p$ is the Gaussian binomial coefficient and $\binom{n}{k}$ is the usual binomial coefficient.
	\item $U$ and $V$ are vector space over $\f_p$. $V\leq U$ if $V$ is a subspace of $U$.
\end{itemize}

\subsection{Linearized polynomial}
Define  \begin{equation}\label{f}f_{c,\vec{a}}(x)=c+a_0x+a_1x^p+a_2x^{p^2}+\cdots+a_hx^{p^h} \end{equation}
be the polynomial over $\f_q$, where $\vec{a}=(a_0,a_1,\cdots,a_h)\in\f_q^{h+1}$ and $c\in\f_q$. Define
\[\N(f_{c,\vec{a}})=\{x\in\f_q|f_{c,\vec{a}}(x)=0\}.\]
When $c=0$, $f_{(0,\vec{a})}$ is a linearized polynomial and $\N(f_{0,\vec{a}})$ is a vector space over $\f_p$. Moreover, $|\N(f_{c,\vec{a}})|=|\N(f_{0,\vec{a}})|$ or $0$. When $|\N(f_{c,\vec{a}})|=|\N(f_{0,\vec{a}})|$, then $\N(f_{c,\vec{a}})=\N(f_{0,\vec{a}})+x_0$, where $c=-f_{0,\vec{a}}(x_0)$ for some $x_0\in\f_q$. 

\subsection{Gaussian binomial coefficients}
We introduce the Gaussian binomial coefficients. For integers $n$ and $k$, the Gaussian binomial coefficient is defined by
$$
\binom{n}{k}_p=\begin{cases}\frac{(p^n-1)(p^{n-1}-1)\cdots(p^{n-k+1}-1)}{(p^k-1)(p^{k-1}-1)\cdots(p-1)},& \text { if } 1 \leq k \leq n ,\\ 1,& \text { if } k=0 ,\\ 0, & \text{ if } k<0 \text{ or } k>n.\end{cases}$$

It is well-known that the number of distinct $k$-dimensional vector subspaces of a given $n$-dimensional vector space over $\f_p$ is  $\binom{n}{k}_p$, and the number of $k$-dimensional vector subspaces of an $n$-dimensional vector space over $\f_p$ containing a fixed $r$-dimensional vector subspace is  $\binom{n-r}{k-r}_p$.

\subsection{$p$-analog M\"{o}bius inversion formula}
	We need the $p$-analog M\"{o}bius inversion formula. Consider $\f_q$ as an $m$-dimensional vector space over $\f_p$, where $q=p^m$. Let $L_m$ denote the posets of all subspaces of $\f_q$, ordered by inclusion. Namely, $V\leq U$ in $L_m$ if $V$ is a subspace of $U$. We have the following lemma.
 \begin{lemma}(\cite{Stanley1997})
 \label{Lemma2.1}
		Let $\mathcal{F}$ and $\mathcal{G}$ be functions defined on $L_m$. Then
		\[\G(U)=\sum_{V\leq U}\F(V)\]
		if and only if 
		\[\F(U)=\sum_{V\leq U}(-1)^{\mathrm{dim} U-\mathrm{dim} V}p^{\binom{\mathrm{dim} U-\mathrm{dim} V}{2}}\G(V),\]
	where $\mathrm{dim} U$ and $\mathrm{dim} V$ denotes the dimensions of $U$ and $V$ over $\f_p$, respectively.
	\end{lemma}

\section{Wight distribution of the extend cyclic code}
\label{sec3}
Let $h$ and $m$ be positive integers with $h<m$ and $q=p^m$ with $p$ a prime. Let $\alpha$ be a fixed primitive element of $\f_q$ and $\alpha_i:=\alpha^i$ for $1\leq i \leq q-1$. Define 	

\begin{equation}
\label{matrix1}
D_h=\left[\begin{matrix}
	1 & 1 & \cdots & 1 & 1 \\
	\alpha_1 & \alpha_2 & \cdots & \alpha_{q-1} & 0 \\
	\alpha_1^p & \alpha_2^p & \cdots & \alpha_{q-1}^p & 0 \\
	\alpha_1^{p^2} & \alpha_2^{p^2} & \cdots & \alpha_{q-1}^{p^2} & 0 \\
	\vdots & \vdots & \ddots & \vdots & \vdots \\
	\alpha_1^{p^h} & \alpha_2^{p^h} & \cdots & \alpha_{q-1}^{p^h} & 0
\end{matrix}\right]
\end{equation}

\noindent Then $D_h$ is an $h+2$ by $q$ matrix. Let $\mathcal{C}_{D_h}$ be the linear code generated by $D_h$ over $\f_q$. It was proved in \cite{Heng2023} that $\mathcal{C}_{D_h}$ is an extended primitive cyclic code. The parameters of $\mathcal{C}_{D_h}$ and the related designs were also studied in \cite{Heng2023}. It was shown that the Hamming weight of any nonzero codeword in $\mathcal{C}_{D_h}$ belongs to the set $\{q-p^h,q-p^{h-1},\cdots, q-p, q-1, q\}$. In the following, we determine the weight distribution of $\mathcal{C}_{D_h}$.

\begin{theorem}\label{TH3.1}Keep the above notation. 
	Then $\mathcal{C}_{D_h}$ is a $[q, h+2, q-p^h]$ code with 
	\[A_{q-p^j}=\binom{m}{j}_pp^{m-j}\sum_{i=0}^{h-j}(-1)^ip^{\frac{i(i-1)}{2}}\binom{m-j}{i}_p(q^{h+1-j-i}-1)\]
for $0\leq j \leq h$ and $A_q=q^{h+2}-1-\sum_{j=0}^hA_{q-p^j}$.
\end{theorem} 
\begin{proof}By the construction of $D_h$, we know that each codeword of  $\mathcal{C}_{D_h}$ corresponds to a polynomial $f_{c,\vec{a}}$ over $\f_q$, namely,
	\[{\mathbf c}_{f_{c,\vec{a}}}=(f_{c,\vec{a}}(\alpha_1),f_{c,\vec{a}}(\alpha_2),\cdots,f_{c,\vec{a}}(\alpha_{q-1}),f_{c,\vec{a}}(0)),\]
	where $f_{c,\vec{a}}$ is defined in (\ref{f}).
	We consider $\f_q$ as an $m$-dimensional vector space over $\f_p$. For $0\leq j \leq h$, the Hamming weight of ${\mathbf c}_{f_{c,\vec{a}}}$ is $q-p^{j}$ if and only if $\N(f_{c,\vec{a}})$ is a $j$-dimensional subspace of $\f_q$ or its coset.
	For any fixed $j$-dimensional $\f_p$-subspace $U$, define	
		\[\F(U)=\{\vec{a}\in\f_q^{h+1}|\N(f_{0,\vec{a}})=U\}.\]
It is easy to see that
\begin{equation}
\label{Aq-pj}
A_{q-p^j}=\binom{m}{j}_pp^{m-j}|\F(U)|,\end{equation}
where $\binom{m}{j}_p$ and $p^{m-j}$ are the number of $j$-dimensional $\f_p$-subspaces of $\f_q$ and the number of cosets of $U$, respectively. Similarly, let
	\[\G(U)=\{\vec{a}\in\f_q^{h+1}|\N(f_{0,\vec{a}})\geq U\}.\]
The cardinality of $\G(U)$ can be determined as follows. Let $\{x_1,x_2,\cdots,x_j\}$ be an $\f_p$-basis of $U$. If $\vec{a}=(a_0,a_1,\cdots,a_h)\in\G(U)$, then the following equation holds.
\begin{equation}\label{matrix2}
	\left[\begin{array}{ccccc}
	x_1 & x_1^p & x_1^{p^2} & \cdots & x_1^{p^h} \\
	x_2 & x_2^p & x_2^{p^2} & \cdots & x_2^{p^h} \\
	\vdots & \vdots & \vdots & \ddots& \vdots \\
	x_j & x_j^p & x_j^{p^2} & \cdots & x_j^{p^h}
\end{array}\right]\left[\begin{array}{c}
	a_0 \\
	a_1 \\
	\vdots \\
	a_h
\end{array}\right]=\left[\begin{array}{c}
	0 \\
	0 \\
	\vdots \\
	0
\end{array}\right].\end{equation}
By Corollary 2.38 in \cite{Lidl1997}, since $x_1,x_2,\cdots,x_j$ are $\f_p$-linearly independent, the coefficient matrix is of full row rank. Then (\ref{matrix2}) has $q^{h+1-j}$ solutions when $0\leq j \leq h$, and has the unique solution $\vec{a}=(0,0,\cdots,0)$ when $j>h$. We have
\begin{equation}
	\label{GU}
|\G(U)|=\begin{cases}q^{h+1-j},& \text { for } 0 \leq j \leq h,\\ 1,& \text { for } j>h.\end{cases}\end{equation}
Moreover, note that $\G(U)=\cup_{V\geq U}\F(V)$ and $\F(V_1)\cap \F(V_2)=\emptyset$ if $V_1\neq V_2$, then
 	\[|\G(U)|=\sum_{V\geq U}|\F(V)|,\]
 or	equivalently, we have 
 	\begin{equation}\label{GU2}	|\G(U)|=\sum_{V^\perp\leq U^\perp}|\F(V)|.\end{equation}
 	Herein and hereafter, $V^{\perp}$ and $U^{\perp}$ are the dual subspaces of $V$ and $U$, respectively.
  	Let $\G'(U^{\perp})=\G(U)$ and $\F'(U^{\perp})=\F(U)$. Then $|\G'(\cdot)|$ and  $|\F'(\cdot)|$ are both functions defined on $L_m$, which is the posets of all $\f_p$-subspaces of $\f_q$. From (\ref{GU2}), we have
  	\begin{equation}\label{GUCZ}
  		|\G'(U^\perp)|=\sum_{V^\perp\leq U^\perp}|\F'(V^\perp)|.
  	\end{equation}
 
 		By (\ref{GUCZ}) and Lemma \ref{Lemma2.1}, we obtain
 		\[|\F'(U^\perp)|=\sum_{V^\perp\leq U^\perp}(-1)^{\mathrm{dim} U^\perp-\mathrm{dim} V^\perp}p^{\binom{\mathrm{dim} U^\perp-\mathrm{dim} V^\perp}{2}}|\G'(V^\perp)|,\]
 i.e.,
 	\[|\F(U)|=\sum_{V\geq U}(-1)^{\mathrm{dim} V-\mathrm{dim} U}p^{\binom{\mathrm{dim} V-\mathrm{dim} U}{2}}|\G(V)|.\]
This with $\sum_{V\geq U}(-1)^{\mathrm{dim} V-\mathrm{dim} U}p^{\binom{\mathrm{dim} V-\mathrm{dim} U}{2}}=0$ leads to 	
\begin{equation}\label{FU}
	|\F(U)|=\sum_{V\geq U}(-1)^{\mathrm{dim} V-\mathrm{dim} U}p^{\binom{\mathrm{dim} V-\mathrm{dim} U}{2}}(|\G(V)|-1).
\end{equation}
		
In the following, we simplify (\ref{FU}). Let $\mathrm{dim} V-\mathrm{dim} U=i$, then $\mathrm{dim} V=i+j$, $0\leq i \leq m-j$. The number of such $V$'s with dimension $i$ is $\binom{m-j}{(i+j)-j}_p=\binom{m-j}{i}_p$. By (\ref{GU}), $|\G(V)|-1=q^{h+1-j-i}-1$ if $i\leq h-j$ and $|\G(V)|-1=0$ otherwise. Then we have	
 \begin{equation}\label{FUJH}
 	|\F(U)|=\sum_{i=0}^{h-j}(-1)^ip^{\frac{i(i-1)}{2}}\binom{m-j}{i}_p(q^{h+1-j-i}-1).\end{equation}
 The desired result follows from (\ref{Aq-pj}) and (\ref{FUJH}). We finish the proof.
\end{proof}

\section{Parameters of designs from $\mathcal{C}_{D_h}$}
\label{sec4}
It was proved  in \cite{Heng2023} that the supports of codewords of any fixed nonzero weight in $\mathcal{C}_{D_h}$ form a $2$-design. More precisely, for a codeword ${\mathbf c}_{f_{c,\vec{a}}}\in\mathcal{C}_{D_h}$, if $\mathrm{wt}({\mathbf c}_{f_{c,\vec{a}}})=q-p^j$ for some $0\leq j \leq h$, then $\mathrm{supp}({\mathbf c}_{f_{c,\vec{a}}})$ is a subset of $\{1,2,\cdots,q\}$ with $q-p^j$ elements. Consequently, $|\mathrm{supp}^c({\mathbf c}_{f_{c,\vec{a}}})|=p^j$, where $\mathrm{supp}^c({\mathbf c}_{f_{c,\vec{a}}})$ denotes the complement set of $\mathrm{supp}({\mathbf c}_{f_{c,\vec{a}}})$ in $\{1,2,\cdots,q\}$. The subset $\mathrm{supp}^c({\mathbf c}_{f_{c,\vec{a}}})$ corresponds to $\N(f_{c,\vec{a}})$, which is an $j$-dimensional subspace of $\f_q$ or one of its cosets. 
For any fixed $j$-dimensional subspace $U$ (respectively, any coset of $U$), Theorem \ref{TH3.1} shows that there are $\sum_{i=0}^{h-j}(-1)^ip^{\frac{i(i-1)}{2}}\binom{m-j}{i}_p(q^{h+1-j-i}-1)$ codewords of $\mathcal{C}_{D_h}$, whose supports are the same. This with the value of $A_{q-p^j}$ leads to
\[|\mathcal{B}_{q-p^j}|=\binom{m}{j}_pp^{m-j}.\]
Since the the pair $(\mathcal{P},\mathcal{B}_{q-p^j})$ forms a $2$-design, we immediately deduce the following theorem by (\ref{eqn-1}).

\begin{theorem}
	For $0\leq j \leq h$, the pair $(\mathcal{P},\mathcal{B}_{q-p^j})$ is a simple $2-(q,q-p^j,\lambda_j)$ design , where 
	\[\lambda_j=(q-p^j-1)\binom{m-1}{j}_p.\]
\end{theorem}

Moreover, it was conjectured that the supports of minimum weight codewords of $\mathcal{C}_{D_h}$ hold $3$-designs when $p=2$. We confirm this conjecture in the following theorem.

\begin{theorem}
	Let $p=2$ and $h\geq 2$. The pair $(\mathcal{P},\mathcal{B}_{q-p^h})$ is a simple $3-(q,q-p^h,\lambda)$ design, where 
	\[\lambda=\frac{\binom{m}{h}_22^{m-h}\binom{q-2^h}{3}}{\binom{q}{3}}.\]
\end{theorem}

\begin{proof}
By Theorem \ref{TH3.1}, we know that $A_{q-2^h}=\binom{m}{h}_22^{m-h}(q-1)$ and 
\[|\mathcal{B}_{q-2^h}|=\frac{1}{q-1}A_{q-2^h}=\binom{m}{h}_22^{m-h}.\]
Consider another block set $\mathcal{B}^c_{q-2^h}=\{\mathrm{supp}^c({\mathbf c})|{\mathbf c}\in\mathcal{C}_{D_h}, \mathrm{wt}({\mathbf c})=q-2^h\}$. Each block in $\mathcal{B}^c_{q-2^h}$ corresponds to a subset of $\f_q$, namely $\N(f_{c,\vec{a}})$, for some $f_{c,\vec{a}}$ with $\mathrm{wt}(f_{c,\vec{a}})=q-2^h$.
We first prove that the pair $(\mathcal{P},\mathcal{B}^c_{q-2^h})$ is a $3$-design. 
For any subset $\{i_1,i_2,i_3\}\subseteq\{1,2,\cdots,q\}$, we know that $\{i_1,i_2,i_3\}$ corresponds to $\{x_{1},x_{2},x_{3}\}$, which is a subset of $\f_q$ with three elements. Let  $<x_1,x_2,x_3>$ be the $\f_2$-subspace generated by $x_1,x_2$, and $x_3$. Note that $\mathrm{dim}<x_1,x_2,x_3>$ is $3$ or $2$ since $p=2$. In the following, we consider the number of $f_{c,\vec{a}}$ such that $\{x_1,x_2,x_3\}\subseteq \N(f_{c,\vec{a}})$.
We discuss in the following two cases.

Case I. $\mathrm{dim}<x_1,x_2,x_3>=3$. It means that $x_1$, $x_2$ and $x_3$ are $\f_2$-linearly independent. Thus, $x_1+x_3$ and $x_2+x_3$ are also  $\f_2$-linearly independent. If $\{x_1,x_2,x_3\}\subseteq \N(f_{c,\vec{a}})$, then $\{x_1+x_3, x_2+x_3\}\subseteq \N(f_{0,\vec{a}})$, hence
 $<x_1+x_3,x_2+x_3>\leq \N(f_{0,\vec{a}})$. This implies that $\N(f_{0,\vec{a}})$ is an $h$-dimensional subspace of $\f_q$, which contains a fixed $2$-dimensional $\f_2$-subspace. Consequently, the number of such $\N(f_{0,\vec{a}})$ is $\binom{m-2}{h-2}_2$. Since $\{x_1,x_2,x_3\}\subseteq \N(f_{c,\vec{a}})$, it remains to make sure that $x_3\in \N(f_{c,\vec{a}})$. Note that there is exact one $c$, namely $c=f_{0,\vec{a}}(x_3)$, such that $x_3\in \N(f_{c,\vec{a}})$. We conclude that there are $\binom{m-2}{h-2}_2$ $f_{c,\vec{a}}$'s such that $\{x_1,x_2,x_3\}\subseteq \N(f_{c,\vec{a}})$ in this case.

Case II.  $\mathrm{dim}<x_1,x_2,x_3>=2$. Without loss of generality, we assume that $x_1$ and $x_2$ are $\f_2$-linearly independent, and then $x_3=0$ or $x_3=x_1+x_2$. Thus, we always have $\{x_1+x_3,x_2+x_3\}=\{x_1,x_2\}$. If $\{x_1,x_2,x_3\}\subseteq \N(f_{c,\vec{a}})$, then $<x_1, x_2>\leq \N(f_{0,\vec{a}})$. The number of such $\N(f_{0,\vec{a}})$ is $\binom{m-2}{h-2}_2$. It remains to consider $x_3\in \N(f_{c,\vec{a}})$. Similarly, there is exact one $c$, namely $c=f_{0,\vec{a}}(x_3)$, such that $x_3\in \N(f_{c,\vec{a}})$. We obtain the same number with that in Case I.

The discussion above shows that the pair $(\mathcal{P}, \mathcal{B}^c_{q-2^h})$ is a $3$-design with parameters $3-(q,2^h,\binom{m-2}{h-2}_2)$. Consequently, $(\mathcal{P}, \mathcal{B}_{q-2^h})$ is a $3$-design as well. By (\ref{eqn-1}), $(\mathcal{P}, \mathcal{B}_{q-2^h})$ is a $3$-$(q,q-2^h,\lambda)$ simple design, where 
\[\lambda=\frac{\binom{m}{h}_22^{m-h}\binom{q-2^h}{3}}{\binom{q}{3}}.\]
The proof is completed.
\end{proof}
\section{Conclusion}
\label{sec5}
In this paper, we studied a family of extended primitive cyclic codes $\C_{D_h}$ and the related $t$-designs. Based on the construction of $\C_{D_h}$, its weight distribution was determined by using the $p$-analog M\"{o}bius inverse formula. Moreover, the parameters of the the related $t$-designs were obtained by the definition of combinatorial designs, which solved a conjecture proposed by Heng et al. in \cite{Heng2023}. We mention that our approach can be used to study other problems involving linearized polynomials.


\begin{thebibliography}{99}	

\bibitem[]{Ding2015}C.~Ding,  
\newblock{Linear codes from some 2-designs},
\newblock \emph{IEEE Trans. Inform. Theory} 61 \penalty0 (6) \penalty0(2015) \penalty0 3265-3275.	

\bibitem[]{Ding2016}C.~Ding, 
\newblock{A construction of binary linear codes from Boolean functions},
\newblock \emph{Discrete Math.} 339 \penalty0 (9) \penalty0 (2016) \penalty0 2288-2303.

\bibitem[]{Ding2018-2}C.~Ding,
\newblock{Infinite families of 3-designs from a type of five-weight code},
\newblock \emph{Des. Codes Cryptogr.} 86 \penalty0 (3) \penalty0 (2018) \penalty0 703-719.

\bibitem[]{Ding2018}C.~Ding,
\newblock{An infinite family of Steiner systems $S(2,4,2^m)$ from cyclic codes},
\newblock \emph{J. Combin. Des.} 26 \penalty0 (3) \penalty0 (2018) \penalty0 127-144.



\bibitem[]{Ding2019}C.~Ding,
\newblock{Designs from linear codes},
\newblock \emph{World Scientific, Hackensack,} \penalty0 2019 \penalty0.


\bibitem[]{Ding2017}C.~{Ding}, C.~{Li}, 
\newblock{Infinite families of 2-designs and 3-designs from linear codes},
\newblock \emph{Discrete Math.} 340 \penalty0 (10) \penalty0 (2017) \penalty0 2415-2431.	


\bibitem[]{Ding2020}C.~{Ding}, C.~{Tang},  
\newblock{Infinite families of near MDS codes holding t-designs},
\newblock \emph{IEEE Trans. Inform. Theory} 66 \penalty0 (9) \penalty0 (2020) \penalty0 5419-5428.	



\bibitem[]{Ding2021}C.~{Ding}, C.~{Tang},
\newblock{The linear codes of t-designs held in the Reed-Muller and simplex codes}, \newblock \emph{Cryptogr. Commun.} 13 \penalty0 (6) \penalty0 (2021) \penalty0 927-949.



\bibitem[]{Dinh2015}H.~{Dinh}, C.~{Li}, Q.~{Yue},
\newblock{Recent progress on weight distributions of cyclic codes over finite fields}, 
\newblock \emph{J. Algebra Comb. Discrete Struct. Appl.} 2 \penalty0 (1) \penalty0 (2015) \penalty0 39-63.



\bibitem[]{Du2020}X.~{Du}, R.~{Wang}, C.~{Fan}, 
\newblock{Infinite families of 2-designs from a class of cyclic codes},
\newblock \emph{J. Combin. Des.} 28 \penalty0 (3) \penalty0 (2020) \penalty0 157-170.





\bibitem[]{Feng2007}K.~{Feng}, J.~{Luo},
\newblock{Value distributions of exponential sums from perfect nonlinear functions and their applications},
\newblock \emph{IEEE Trans. Inform. Theory} 53 \penalty0 (9) \penalty0  (2007) \penalty0 3035-3041.




\bibitem[]{Heng2020}Z.~{Heng}, Q.~{Wang}, C.~{Ding},
\newblock{Two families of optimal linear codes and their subfield codes},
\newblock \emph{IEEE Trans. Inform. Theory} 66 \penalty0 (11) \penalty0 (2020) \penalty0 6872-6883.




\bibitem[]{Heng2023}Z.~{Heng}, X.~{Wang}, X.~{Li},
\newblock{Constructions of cyclic codes and extended primitive cyclic codes with their applications},
\newblock \emph{Finite Fields Appl.} 89 \penalty0  (2023) \penalty0 102208.


\bibitem[]{Heng2016}Z.~{Heng}, Q.~{Yue},
\newblock{Several classes of cyclic codes with either optimal three weights or a few weights},
\newblock \emph{IEEE Trans. Inform. Theory} 62 \penalty0 (8) \penalty0 (2016) \penalty0  4501-4513.	




\bibitem[]{PLess2003}W.C.~{Huffman}, V.~{Pless}, 
\newblock{Fundamentals of Error-Correcting Codes},
\newblock \emph{Cambridge University Press,}  \penalty0 Cambridge, \penalty0 2003.


\bibitem[]{Li2014}C.~{Li}, Q.~{Yue}, F.~{Li}, 
\newblock{Hamming weights of the duals of cyclic codes with two zeros},
\newblock \emph{IEEE Trans. Inform. Theory} 60 \penalty0 (7) \penalty0 (2014) \penalty0 3895-3902.


\bibitem[]{Lidl1997}R.~{Lidl}, H.~{Niederreiter},
\newblock{Finite Fields: 2nd ed.},
\newblock \emph{Cambridge University Press,} \penalty0 Cambridge, \penalty0 1997.		



\bibitem[]{Liu2021}Y.~{Liu}, X.~{Cao}, 
\newblock{Infinite families of 2-designs derived from affine-invariant codes},
\newblock \emph{J. Combin. Des.} 29 \penalty0 (10) \penalty0 (2021) \penalty0 683-702.



\bibitem[]{Liu2022}Y.~{Liu}, X.~{Cao}, 
\newblock{A class of affine-invariant codes and their support 2-designs},
\newblock \emph{Cryptogr. Commun.} 14 \penalty0 (2) \penalty0 (2022) \penalty0 215-227.




\bibitem[]{Luo2008}J.~{Luo}, K.~{Feng}, 
\newblock{On the weight distributions of two classes of cyclic codes},
\newblock \emph{IEEE Trans. Inform. Theory} 54 \penalty0 (12) \penalty0 (2008) \penalty0 5332-5344.



\bibitem[]{Stanley1997}R.~{Stanley}
\newblock{Enumerative Combinatorics, vol. I},
\newblock \emph{Cambridge University Press,} \penalty0 Cambridge, \penalty0 New York, \penalty0 1997.



\bibitem[]{Tang2021}C.~{Tang}, C.~{Ding}, 
\newblock{An infinite family of linear codes supporting 4-designs},
\newblock \emph{IEEE Trans. Inform. Theory} 67 \penalty0 (1) \penalty0 (2021) \penalty0 244-254.



\bibitem[]{Tang2019}C.~{Tang}, C.~{Ding}, M.~{Xiong},
\newblock{Steiner systems $S(2,4,\frac{3^m-1}{2})$ and 2-designs from ternary linear codes of length $\frac{3^m-1}{2}$},
\newblock \emph{Des. Codes Cryptogr.} 87 \penalty0 (12) \penalty0 (2019) \penalty0 2793-2811.



\bibitem[]{Wang2023}X.~{Wang}, C.~{Tang}, C.~{Ding}, 
\newblock{Infinite families of cyclic and negacyclic codes supporting 3-designs},
\newblock \emph{IEEE Trans. Inform. Theory} 69 \penalty0 (4) \penalty0 (2023) \penalty0 2341-2354.



\bibitem[]{Xiang2022}C.~{Xiang},
\newblock{Some t-designs from BCH codes},
\newblock \emph{Cryptogr. Commun.} 14 \penalty0 (3) \penalty0 (2022) \penalty0 641-652.



\bibitem[]{Xiang2020}C.~{Xiang}, X.~{Ling}, Q.~{Wang}, 
\newblock{Combinatorial t-designs from quadratic functions},
\newblock \emph{Des. Codes Cryptogr.} 88 \penalty0 (3) \penalty0 (2020) \penalty0 553–565.



\bibitem[]{Xiang2022-2}C.~{Xiang}, C.~{Tang}, Q.~{Liu}, 
\newblock{An infinite family of antiprimitive cyclic codes supporting Steiner systems S(3,8,$7^m+1$)},
\newblock \emph{Des. Codes Cryptogr.} 90 \penalty0 (6) \penalty0 (2022) \penalty0 1319-1333.



\bibitem[]{Xiong2016}M.~{Xiong},N.~{Li},Z.~{Zhou},C.~{Ding},
\newblock{Weight distribution of cyclic codes with arbitrary number of generalized Niho type zeroes,}
\newblock \emph{Des. Codes Cryptogr.} 78 \penalty0 (3) \penalty0 (2016) \penalty0 713-730.



\bibitem[]{Xu2022}G.~{Xu}, X.~{Cao}, L.~{Qu},
\newblock{Infinite families of 3-designs and 2-designs from almost MDS codes},
\newblock \emph{IEEE Trans. Inform. Theory} 68 \penalty0 (7) \penalty0 (2022) \penalty0 4344-4353.



\bibitem[]{Yan2022}Q.~{Yan}, J.~{Zhou},
\newblock{Infinite families of linear codes supporting more t-designs},
\newblock \emph{IEEE Trans. Inform. Theory} 68 \penalty0 (7) \penalty0 (2022) \penalty0 4365-4377.



\bibitem[]{Zeng2010}X.~{Zeng}, L.~{Hu}, W.~{Jiang}, Q.~{Yue}, X.~{Cao},
\newblock{The weight distribution of a class of p-ary cyclic codes},
\newblock \emph{Finite Fields Appl.} 16 \penalty0 (1) \penalty0 (2010) \penalty0 56-73.



\bibitem[]{Zhou2014}Z.~{Zhou}, C.~{Ding},
\newblock{A class of three-weight cyclic codes},
\newblock \emph{Finite Fields Appl.} 25 \penalty0 (2014) \penalty0 79-93.
	\end{thebibliography}
\end{document}